\documentclass[reprint, aps, pra, floatfix, superscriptaddress, nofootinbib]{revtex4-1}
\usepackage{graphicx}    
\usepackage{dcolumn}    
\usepackage{bm}             
\usepackage{amssymb}   
\usepackage{amsmath}    
\usepackage{amsthm}
\usepackage{mathrsfs}    
\usepackage{commath}   
\usepackage{subfigure}    
\usepackage{braket}
\usepackage{natbib}
\usepackage{float}
\usepackage{color}
\usepackage{siunitx}
\usepackage[colorlinks=true, allcolors=blue]{hyperref}
\usepackage{hyphenat}
\usepackage{notes2bib}

\graphicspath{{Figures/}}

\newtheorem{theorem}{Theorem}
\newtheorem{lemma}{Lemma}

\begin{document}

\title{Entanglement trimming in stabilizer formalism}

\author{Changchun Zhong}
\affiliation{Pritzker School of Molecular Engineering, University of Chicago, Chicago, IL 60637, USA}

\author{Yat Wong}
\affiliation{Pritzker School of Molecular Engineering, University of Chicago, Chicago, IL 60637, USA}

\author{Liang Jiang}
\affiliation{Pritzker School of Molecular Engineering, University of Chicago, Chicago, IL 60637, USA}

\date{\today}

\begin{abstract}

Suppose in a quantum network, there are $n$ qubits hold by Alice, Bob and Charlie, denoted by systems $A$, $B$ and $C$, respectively. We require the qubits to be described by a stabilizer state and assume the system $A$ is entangled with the combined system $BC$. An interesting question to ask is when it is possible to transfer all the entanglement to system $A$ and $B$ by local operation on $C$ and classical communication to $AB$, namely \textit{entanglement trimming}. We find a necessary and sufficient condition and prove constructively for this entanglement trimming, which we name it as ``the bigger man principle". This principle is then extended to qudit with square-free dimension and continuous variable stabilizer states. 

\end{abstract}

\maketitle

\textit{Introduction}.---Distributed quantum architecture has been recently proposed for various quantum enhanced applications, e.g., quantum metrology, quantum communication and quantum computation \cite{wehner2018,sansavini2020,zhuang2020,Nielsen2002}. These quantum enhancements usually require efficiently generating and distributing entanglement between different quantum nodes in a network through connected quantum channels \cite{Kimble2008}. In practice, the quantum entanglement in the network will suffer from environment noise, leading to a rapid resource decay in time. Thus an efficient algorithm to manipulate the entanglement is critical. To get a higher two-node entanglement, one can always resort directly to the connected quantum channels. However, in practice the quantum channel could be very expensive. A possible more economical resolution is to make full use of the preexisting entanglement in the network and perform some local operations on other nodes such that the entanglement could concentrate to the chosen two nodes \cite{Loock2000,Jing2006}. Similar question has been raised up in Ref.~\cite{stephanie2020,pant2019,pant2019_}, where Bell state between chosen nodes is targeted through local operations. In this paper, we study a tripartite network by asking the question: when and how two parties can preserve all the entanglement if the third party is removed. This problem can be captured by a simple model: as illustrated in Fig.~\ref{fig1}, for a tripartite quantum system described by a density operator $\rho^{ABC}$, whether it is possible to locally disentangle the system $C$ while transfer the entanglement between $A$ and $BC$ to system $A$ and $B$, namely the \textit{entanglement trimming}: $\rho^{ABC}\rightarrow \Tilde{\rho}^{AB}\otimes\Tilde{\rho}^C$ with $E(\rho^{A|BC})=E(\Tilde{\rho}^{A|B})$ ($E$ is any entanglement measure).

In particular, we study the entanglement trimming of a given stabilizer state with tri-partition using only single party Clifford operations and classical communication. We answer the question in the positive by showing the possibility of precisely transferring the tripartite entanglement to the chosen two parties through local operations on the third. 
Unlike general entangled states, which need exponential large number of parameters to describe, the stabilizer states can be efficiently specified by its stabilizer generators, e.g., only $2n^2$ bits are needed to describe a $n$-qubit stabilizer state, and the state property is fully determined by its generators \cite{Daniel1997,Daniel1998}. Through exploring the stabilizer generator structures, e.g., the canonical decomposition \cite{Fattal2004}, we identity a necessary and sufficient condition for the entanglement trimming, and we call it ``\textit{the bigger man principle}"---to be respectful to information theorists and their preference of using metaphor to reveal entanglement rules, e.g., entanglement monogamy \cite{Coffman2000,Terhal2004}, and here we refer the bigger man to Charlie (system $C$), who sacrifices himself in ending a triangular relation and remarkably not hurting others. Although the stabilizer states, together with Clifford operations are not enough to demonstrate quantum advantage \cite{Daniel1996}, they contain rather rich quantum structures and play an essential role in the measurement based quantum computation \cite{raussendorf2001,Hein2004}. Thus revealing their entanglement properties would be helpful in designing new quantum algorithms. In this paper, we first discuss ``the bigger man principle" of entanglement trimming for qubit stabilizer states, then extend it to the stabilizer state of general discrete and continuous variables.

\begin{figure}[t]
\includegraphics[width=\columnwidth]{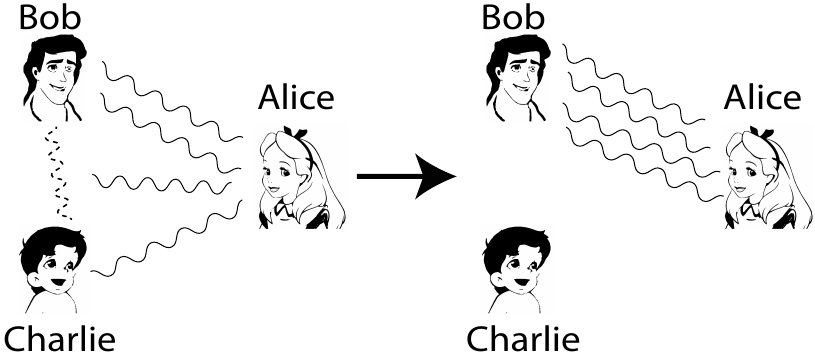}
\caption{Alice is initially entangled with the combined system hold by Bob and Charlie. Local measurement on Charlie's system gets him disentangled, while all entanglement is transferred to between Alice and Bob. \label{fig1}}
\end{figure}

\textit{Entanglement trimming for qubit stabilizer state}.---The stabilizer formalism is well known as a powerful tool in discussing a wide class of quantum error correcting codes \cite{Daniel1996,Daniel19982} as well as analyzing the complexity of classical simulation of quantum computations \cite{Daniel1997}. The stabilizer for a quantum error correcting code is an Abelian subgroup of the Pauli group for $n$ qubits. Any state in the code is a simultaneously eigen-state of the Abelian subgroup elements. If the subgroup has $k$ generators, it defines an error correcting code with dimension $2^{n-k}$. When $k=n$, the stabilizer uniquely determines a stabilizer state. Stabilizer states can be efficiently described by their group generators, which scales only polynomially with the number of qubits. By exploring the structure of the group generators, many fruitful results have been achieved, such as quantifying stabilizer state entanglement \cite{Fattal2004,Koenraad2005}, developing the Gottesman-Knill theorem \cite{Gottesman1998,Stephen2002}, etc. Suppose we have a stabilizer $S=\braket{g_1,g_2,...g_n}$ with $n$ generators, which stabilizes an $n$-qubit state
\begin{equation}\label{eqs1}
    \rho=\frac{1}{2^n}\sum_{g\in S}g.
\end{equation}
With a partition into system $A$ and $B$ ($n_A+n_B=n$), the stabilizers in $S$ can be divided into several subgroups $S_{AB}=\braket{g_i^A\otimes g_i^B}_{i=1,...,\abs{S_{AB}}}$, $S_A=\braket{g_j^A\otimes I^B}_{j=1,...,\abs{S_{A}}}$ and $S_B=\braket{I^A\otimes g_k^B}_{k=1,...,\abs{S_{B}}}$, where $\abs{\cdot}$ denotes the group rank. $S_{AB}$ is responsible for the state correlation while $S_A$ and $S_A$ give the local information. The bipartite entanglement can be obtained by calculating the von Neumann entropy of the reduced density operator. Since Pauli operators are trace free, we have
\begin{equation}
    \rho_B=\frac{1}{2^{n_B}}\sum_{g\in S_B}g=\frac{1}{2^{n_B-\abs{S_B}}}\prod_{k=1}^{\abs{S_B}}\frac{I+g^B_k}{2}.
\end{equation}
where $\prod_{k=1}^{\abs{S_B}}\frac{I+g^B_k}{2}$ is the projection onto a subspace with dimension $2^{n_B-\abs{S_B}}$. The von Neumann entropy is thus evaluated to be $E=-\text{tr}(\rho_B\log_2\rho_B)=n_B-\abs{S_B}$. Similarly, by taking partial trace over system $B$, we can get another expression of the entanglement as $E=n_A-\abs{S_A}$. Since $\rho$ is a pure state, we have $n_A-\abs{S_A}=n_B-\abs{S_B}$. Combining the fact that $\abs{S_A}+\abs{S_B}+\abs{S_{AB}}=n$, it is easy to get that the entanglement $E={\abs{S_{AB}}}/{2}$,
where we see that the number of generators for the subgroup $S_{AB}$ determines the entanglement. It is in this sense that the generators for $S_{AB}$ are called correlated, while the generators for $S_A$ and $S_B$ are named local.

\begin{lemma}\label{the1}---\textbf{Canonical decomposition of generators for qubit}: the generator $S=\braket{g_1,g_2,...,g_n}$ for $n$-qubit stabilizer state of any bi-partition $A|B$ can always be brought into the canonical form: $S=\braket{a_i\otimes I_B,I_A\otimes b_j,g_k^A\otimes g_k^B,\Bar{g}_k^A\otimes\Bar{g}_k^B}$, where the first two subsets generate $S_A$ and $S_B$, and the last two sets generate $S_{AB}$, responsible for correlation. The generators of $S_{AB}$ collect into $E=\abs{S_{AB}}/2$ pairs $\{g_k\equiv g_k^A\otimes g_k^B,\Bar{g}_k\equiv \Bar{g}_k^A\otimes\Bar{g}_k^B\}$, where the opeartor $g_k^A$ $(g_k^B)$ commutes with all canonical generators of $S$ except $\Bar{g}_k$, while the operator $\bar{g}_k^A$ $(\bar{g}_k^B)$ commutes with all generators in $S$ except $g_k$---we define this property as \textbf{exclusive commutation relation}. 
\end{lemma} 
The proof of this lemma is first given in Ref.~\cite{Fattal2004}, later we will generalize the result to qudit with prime dimension and continuous variable stabilizer state. It is worth mentioning that identifying the canonical form of a given stabilizer can be done efficiently, which scales polynomially in the number of qubits \cite{Fattal2004,Koenraad2005}.
This lemma gives us an easy way of getting the bi-partite entanglement by counting the $E$ correlated pairs. More importantly it allows us to identify the set of stabilizer states that are possible for entanglement trimming and also gives us a constructive way to realize it. 

Suppose the $n$-qubit stabilizer state, stabilized by $S=\braket{g_1,g_2,...g_n}$, is hold by three parties $A$, $B$ and $C$. If taking the partition $AB|C$ (group $AB$ together), then according to lemma \ref{the1}, we can bring the stabilizer generators into the canonical form, \begin{equation}\label{ecano}
    S=\braket{a_i^{AB}\otimes I^C,I^{AB}\otimes b_j^C,g_k^{AB}\otimes g_k^{C},\bar{g}_k^{AB}\otimes\bar{g}_k^{C}},
\end{equation}
where the last two sets of generators determine the $AB|C$ entanglement. 
\begin{theorem}
---\textbf{the bigger man principle for qubit stabilizer state}: given a general $n$-qubit stabilizer state $\rho^{ABC}$ stabilized by Eq.~(\ref{ecano}), the $A|BC$ entanglement can be transferred to $A|B$ by local operations on $C$ if and only if, (1) for any given $k$, one of its paired generators $\{g_k^{AB}\otimes g_k^C,\bar{g}_k^{AB}\otimes\bar{g}_k^C\}$ can be expressed as local operator in terms of $A|BC$, e.g., $g_k^{AB}g_k^C=I^Ag_k^Bg_k^C$, or (2) even number of pairs of correlated generators can be expressed in the form, e.g., $\{I^Ag_{k^\prime}^B\otimes g_{k^\prime}^C,I^A\bar{g}_{k^\prime}^B\otimes \bar{g}_{k^\prime}^C\}$ and $\{{g}_k^AI^B\otimes g_k^C,\bar{g}_k^AI^B\otimes \bar{g}_k^C\}$.
\end{theorem}
\begin{proof}
If part: Suppose we have $\bar{g}^{AB}_k=I^A\otimes \bar{g}^B_k$ for any $k$ satisfying condition (1), the initial stabilizer is 
\begin{equation}\label{eqbs}
S_i=\braket{a_i^{AB}\otimes I^C,I^{AB}\otimes b_j^C, g_k^{AB}\otimes g_k^C,I^A \bar{g}_k^B\otimes\bar{g}_k^C}.
\end{equation}
To do entanglement trimming, we need to disentangle $C$ from the whole system. It can be realized by local measurement on $C$. Noticing the operator $g^C_k$ commutes all generator in $S_i$ except $I^A\bar{g}^B_k\otimes\bar{g}^C_k$, we thus can measure $g^C_k$, and a renewed stabilizer can be obtained (assuming measurement result $+1$ in all later discussions)
\begin{equation}
    S_f=\braket{a_i^{AB}\otimes I^C,I^{AB}\otimes b_j^C,g_k^{AB}\otimes I^C,I^{AB}\otimes g_k^{C}},
\end{equation}
where we see system $C$ becomes separable with $AB$ and the $A|B$ entanglement $E_f(A|B)$ is totally determined by those generators $\{a_i^{AB},g_k^{AB}\}$. Also, we know the $A|BC$ entanglement $E_i(A|BC)$ in $S_i$ is captured by the generators $\{a_i^{AB}\otimes I^C,g_k^{AB}\otimes g_k^C\}$ \cite{snotea}. Obviously, we have $E_i(A|BC)=E_f(A|B)$. 

If $S_i$ has some generators in the form of condition (2), the generator pairs can be rewritten as $\{I^Ag_{k^\prime}^B\otimes {g}_{k^\prime}^C,\bar{g}_k^A\bar{g}_{k^\prime}^B\otimes \bar{g}_k^C\bar{g}_{k^\prime}^C\}$ and $\{I^A\bar{g}_{k^\prime}^B\otimes \bar{g}_{k^\prime}^C,g_k^Ag_{k^\prime}^B\otimes g_k^Cg_{k^\prime}^C\}$, where we see, in each new paired generators, one is in the $A|BC$ local form. Condition (2) is thus reduced to the form of condition (1) with the only difference that the first generators in each new pair do not satisfy the exclusive commutation relation, specifically, $[I^Ag^B_{k^\prime},I^A\bar{g}^B_{k^\prime}]\neq 0$. Luckily, commutation of the two operators is not required since we can perform similar measurements on $C$, e.g., the operators $\bar{g}_k^C\bar{g}_{k^\prime}^C$, ${g}_k^C{g}_{k^\prime}^C$, to disentangle $C$, realizing the required entanglement trimming. This fact also applies in the following ``only if" proof.

Only if part: We first transform the stabilizer state into a product of un-entangled qubit states in system $C$ and a state with the remaining qubits maximally entangled with $AB$. As shown in the supplement \cite{supp}, this transformation can be accomplished by local Clifford gates on $C$ only, which obviously don't change entanglement. Thus we write down the stabilizer in a simpler canonical form (neglecting those unentangled qubits, $C$ is maximally entangled with $AB$)
\begin{equation}
    S_i=\braket{a_i^{AB}\otimes I^C,g_k^{AB}\otimes g_k^{C},\bar{g}_k^{AB}\otimes\bar{g}_k^{C}}.
\end{equation}
We then need to perform measurement on system $C$. If not commuting all generators in $S_i$, a Pauli measurement operator can always be made anti-commuting with only one generator in $S_i$ \cite{Nielsen2002}. Let's first assume the operator $g^C_{k=1}$ is measured on $C$, which only anti-commutes the generator $\bar{g}^{AB}_1\otimes\bar{g}^C_1$. The measurement leads to a generator replacement $\bar{g}^{AB}_1\otimes\bar{g}^C_1\rightarrow I^{AB}\otimes g^C_1$, where the new generator is local in terms of $AB|C$ (making $C$ not maximally entangled again) and can be similarly removed by local Clifford operations on $C$. Then Up to a generator multiplication, a new stabilizer can be obtained $\braket{a_i^{AB}\otimes I^C,g_1^{AB}\otimes I^{C},g_k^{AB}\otimes g_k^{C},\bar{g}_k^{AB}\otimes\bar{g}_k^{C}}_{k\neq1}$ \cite{nob}. This process can be repeated for all $k$ until we totally separate the system $C$
\begin{equation}
    S_f=\braket{a_i^{AB}\otimes I^C,g_k^{AB}\otimes I^{C}}.
\end{equation}
Comparing $S_i$ and $S_f$, it is obvious $\bar{g}^{AB}_k=I^A\bar{g}^B_k$ is needed for the equivalence $E_i(A|BC)=E_f(A|B)$.

Now a question to ask is whether a better Pauli measurement than $g^C_k$ exists. We now show the answer is negative. Let's assume a Pauli operator $p^C_1\neq g^C_{k=1}$ which only anti-commutes $\bar{g}^{AB}_1\otimes\bar{g}^C_1$ in $S_i$. The $p^C_1$ measurement and similar procedure yield a stabilizer $\braket{a_i^{AB}\otimes I^C,g_1^{AB}\otimes g_1^{C}p_1^C,g_k^{AB}\otimes g_k^{C},\bar{g}_k^{AB}\otimes\bar{g}_k^{C}}_{k\neq1}$, where $g^C_1p^C_1$ is not identity. Nevertheless, up to a correlated generator multiplication, the transform $g_1^{AB}\otimes g^C_1p^C_1\rightarrow \Tilde{g}_1^{AB}\otimes I^C$ must exists since the correlated generators must be paired in the canonical form \cite{supp}. However, $\Tilde{g}_1^{AB}$ is not guaranteed to be in a correlated form in terms of $A|B$, e.g., we might have $\Tilde{g}_1^{AB}=I^A\Tilde{g}^B_1$, reducing the entanglement that can be trimmed to $A|B$. 
\end{proof}


To better convey the idea, we now show some examples. First, for a GHZ state, stabilized by the generators $S=\braket{Z_AZ_BI_C,I_AZ_BZ_C,X_AX_BX_C}$, each party holds one of the three qubits. The entanglement is $1$ ebit between $A$ and $BC$. After a measurement of operator $I_AI_BX_C$, which anti-commutes $I_AZ_BZ_C$ and commutes with other generators in $S$, (also assuming $+1$ eigenvalue is obtained), we get a renewed stabilizer $S^\prime=\braket{Z_AZ_BI_C,I_AI_BX_C,{X_AX_BI_C}}$, where we see that $C$ is disentangled, and $A|B$ inherits the $1$ ebit entanglement from $A$ and $BC$. This is possible since we initially have a generator $I_A Z_BZ_C$, satisfying the trimming condition. Notice the GHZ state has a symmetric structure, which means we can perform measurement on any one of the three parties and preserve the entanglement in the remaining two. As a second example, for a tripartite four-qubit state with stabilizer $S=\braket{X_AI_BX_CI_C,Z_A I_BZ_CI_C,I_AZ_BI_CZ_C,I_AX_BI_CX_C}$, where $A|BC$ is entangled with $1$ ebit, and the two pairs $\{X_AI_B,Z_AI_B\}$ and $\{I_AX_B,I_AZ_B\}$ satisfy entanglement trimming condition (2). Thus we can rewrite the stabilizer $S=\braket{X_AZ_BX_CZ_C,I_AX_BI_CX_C,Z_AX_BZ_CX_C,I_AZ_BI_CZ_C}$. After the Pauli measurements, we disentangle the system $C$ and get a stabilizer state for $A|B$, stabilized by $\braket{X_AZ_B,Z_AX_B}$ preserving the $1$ ebit entanglement.

\textit{Qudit entanglement trimming}.---The bigger man principle can be extended to qudit stabilizer states. Analogously, one can define an $n$-qudit stabilizer state as
\begin{equation}\label{st1}
    \rho=\frac{1}{D^n}\sum_{g\in S}g,
\end{equation}
where $D$ is the qudit dimension, the stabilizer $S=\braket{g_1,g_2,...}$ and $g_i$ is a generalized Pauli generators for $n$ qudits \cite{Erik2005,Hayashi2017,Hayashi_2017}. For any generalized Pauli group defined on $n$ qudits, the group elements satisfy the relation $g_ig_j=g_jg_i\omega^{\lambda_{ij}}$,
where $\omega=e^{i2\pi/D}$ and $\lambda_{ij}$ is an integer in a ring $Z_D$. Equation (\ref{st1}) represents a stabilizer state since we have $S\rho=\rho$, and obviously it respects the condition tr$\rho=1$. In qubit case, whenever the number of stabilizer generator equals the number of qubits, its stabilizer state will be a pure state. However, for qudits, it is different since in general we have $tr\rho^2=\frac{\|S\|}{D^n}$, where the symbol $\|\cdot\|$ denotes the group size (note we use $\abs{\cdot}$ for group rank).
We see that $tr\rho^2=1$ if and only if $\|S\|=D^n$. In general for qudit, the size of a group with $n$ generators is not equal to $D^n$. Actually we have $\|S\|=\prod_i^nO_i\le D^n$,
where $O_i$ is the order of the group generator $g_i$ and it takes the value of either $D$ or the division of $D$. Obviously, a sufficient condition for $\rho$ to be pure is $D$ being a prime number, while for general $D$, $\rho$ defined by Eq. (\ref{st1}) with only $n$ generators is a mixed state \cite{qudits,Hayashi_2017,Vlad2011,Erik2005}. To have a pure state, we confine the discussion to prime dimension (later, we will extend the discussion to composite square free dimension \cite{enotc}). Similar to the qubit case, to evaluate the entanglement, we take the partition $A|B$ with $n_A+n_B=n$. The stabilizer is decomposed into the local subgroups $S_A=\braket{g_i^A\otimes I^B}_{i=1,...,\abs{S_A}}$, $S_B=\braket{I^A\otimes g_j^B}_{j=1,...,\abs{S_B}}$ and the correlated subgroup $S_{AB}=\braket{g_k^A\otimes g_k^B}_{k=1,...,\abs{S_{AB}}}$ with $\abs{S_A}+\abs{S_B}+\abs{S_{AB}}=n$.
By tracing out the sytem $A$, we have a reduced density matrix
\begin{equation}
    \rho_B=\frac{1}{D^{n_B}}\sum_{g\in S_B}g.
\end{equation}
Noticing $\frac{1}{\|{S_B}\|}\sum_{g\in S_B}g$ is a projector onto the subspace with dimension $D^{n_B}/||S_B||$, the von Neumann entropy is thus evaluated to be
$ E=n_B\log_2D-\log_2\|{S_B}\|$.
Since $\rho$ is pure, we also have
$E=n_A\log_2D-\log_2\|{S_A}\|$ by tracing over system $B$.
Combining the two expressions, we find $E=\frac{1}{2}\log_2\|S_{AB}\|$.
Notice $D$ is prime, the entanglement further takes the form $E=\frac{1}{2}\abs{S_{AB}}\log_2D$, obviously reducing to the qubit case when $D=2$. The result indicates the bipartite entanglement of a qudit stabilizer state is also determined by the number of correlated generators in the stabilizer group.

\begin{lemma}\label{th3}
---\textbf{Canonical decomposition for qudit}:
Denote a stabilizer $S=\braket{g_1,g_2,...,g_n}$ with $n$ generators for a pure $n$-qudit stabilizer state, then for any bi-partition $A|B$, its generators can always be brought into the canonical form:
\begin{equation}
    S=\braket{a_i\otimes I_B,I_A\otimes b_j,g_k^A\otimes g_k^B,\Bar{g}_k^A\otimes\Bar{g}_k^B}
\end{equation}
where the first two subsets generate $S_A$ and $S_B$, responsible for local information, and the last two generate $S_{AB}$, responsible for correlation. The generators of $S_{AB}$ collect into $E=\abs{S_{AB}}/2$ pairs $\{g_k=g_k^A\otimes g_k^B,\Bar{g}_k=\Bar{g}_k^A\otimes\Bar{g}_k^B\}$, where $g_k^A$ $(g_k^B)$ commutes with all canonical generators of $S$ except $\Bar{g}_k$, and $\bar{g}_k^A$ $(\bar{g}_k^B)$ commutes with all canonical generators of $S$ except $g_k$.
\end{lemma} 
As in the qubit case \cite{Fattal2004}, the proof of this lemma explores the transformation of group generators,
and the detail is given in the supplementary material \cite{supp}. By analyzing the correlated generators in the decomposition, it turns out that a similar entanglement trimming theorem as in qubit case can be constructed in a similar way. The differences are: when we perform a generalized Pauli measurement, the result could be $\omega^\lambda$ with $\lambda=0,1,...,D-1$, and we should replace the stabilizer elements accordingly \cite{Daniel1998}; the generalized Pauli generators satisfy $g^D=I$, which requires us to take the multiplication several times to get identity operator when we renew the generators. We summarize the theorem in the following:
\begin{theorem}
---\textbf{the bigger man principle for qudit stabilizer state}: given an $n$-qudit stabilizer state with tri-partition $\rho^{ABC}$, and its stabilizer generators are in the canonical form $S=\braket{a_i^{AB}\otimes I^C,I^{AB}\otimes b_j^C,g_k^{AB}\otimes g_k^{C},\bar{g}_k^{AB}\otimes\bar{g}_k^{C}}$. The $A|BC$ entanglement can be reduced to $A|B$ by local $C$ measurement if and only if, (1) for given $k$, one of the paired operators is local, e.g., $g_k^{AB}=I^Ag_k^B$, or (2) there are pairs of operators in the form, e.g., $\{I^Ag_k^B\otimes g_k^C,I^A\bar{g}_k^B\otimes \bar{g}_k^C,g^A_{k^\prime}I^B\otimes g_{k^\prime}^C,\bar{g}^A_{k^\prime}I^B\otimes \bar{g}_{k^\prime}^C\}$.  
\end{theorem}

Besides qudit with prime dimension, the bigger man principle can be applied to qudit stabilizer states with composite square-free dimension, which is a direct result of the Chinese remainder isomorphism. As first discussed in Ref.~\cite{Shiang2011}, any qudit stabilizer state of composite dimension can be decomposed into tensor product of several stabilizer states with prime power dimension. Thus for a square-free dimensional stabilizer state, the entanglement trimming can still be done within the above formalism if each component in the tensor product satisfies the trimming condition. 
Interesting, when the condition is not satisfied, the entanglement trimming is still possible but can no longer be described by stabilizer formalism (see Ref.~\cite{supp} for more details). Furthermore, the principle is also true for continuous variable stabilizer state. As an extension of discrete variable Pauli group, the generalized Pauli group for continuous variables is called Heisenberg-Weyl (HW) group, composed of all phase space displacements \cite{Richard2004,Kok2010}. Interestingly, one can build up a correspondence between the HW group and a symplectic vector space, and each group element is represented by a symplectic vector. Moreover, an Abelian subgroup corresponds to a vector subspace with all elements being symplectic orthogonal, which defines a Lagrangian plane \cite{Gosson2006}. This vector subspace, rather than the commuting displacement operators, enables us to concisely formulate the canonical decomposition of continuous variable stabilizer generators. Detailed in the supplement \cite{supp}, we show that the vector subspace also wears a canonical form, with which continuous variable bigger man principle can be established.

\textit{Outlook}.---There are several interesting questions to be answered in the future. First, we studied the tripartite entanglement trimming in this paper, it would be worthwhile to explore the trimming condition for more parties, which is more practical in real quantum networks. Also, The principle we discussed is based on pure stabilizer states. While interesting as it is, we might often encounter mixed entangled state. In fact, we can define a mixed stabilizer state similarly \cite{Koenraad2005}. As we mentioned before, the stabilizer state defined in Eq. (\ref{eqs1}), or Eq. (\ref{st1}) will be mixed when the number of generator is not chosen properly compared to the number of qubits (qudits). In these cases, the same canonical decomposition of the stabilizer generators can be obtained, and we can claim a same ``bigger man principle" if we only require von Neumann entropy to be equal before and after the local opeartion on system $C$. However, since the von Neumann entropy for the reduced density operator of a mixed state can't faithfully measure the entanglement, it could be more involving to get a similar ``bigger man principle", e.g., possibly introducing other mixed state entanglement measure \cite{Charles1996}, which is an interesting and important research topic. Furthermore, the stabilizer state, although significant in the quantum application, forms only a subset of all quantum states. The search for the bigger man principle of general entangled state should be an exciting work to do in the future.

\begin{acknowledgments}
 C.Z. thanks Filip Rozpedek and Wenlong Ma for helpful discussions. We acknowledge support from the ARO (W911NF-18-1-0020, W911NF-18-1-0212), ARO MURI (W911NF-16-1-0349), AFOSR MURI (FA9550-19-1-0399), NSF (EFMA-1640959, OMA-1936118, EEC-1941583), NTT Research, and the Packard Foundation (2013-39273). 
\end{acknowledgments}

\bibliographystyle{apsrev4-1}
\bibliography{err}

\begin{thebibliography}{39}
\bibitem [{sno()}]{snotea}%
  \BibitemOpen
  \href@noop {} {}\bibinfo {note} {The operator $g^{AB}_k$ can't take the form
  $g^A_k\otimes I^B$ since it should anti-commute with
  $\bar{g}^{AB}_k=I^A\otimes\bar{g}^B_k$ due to the exclusive commutation
  relation.}\BibitemShut {Stop}%
\bibitem [{nob()}]{nob}%
  \BibitemOpen
  \href@noop {} {}\bibinfo {note} {Rigorously speaking, all generators must be
  renewed by the local Clifford gates on $C$. Since they don't change the
  stabilizer structures, we just use the same symbols for
  simplicity.}\BibitemShut {Stop}%
\bibitem [{qud()}]{qudits}%
  \BibitemOpen
  \href@noop {} {}\bibinfo {note} {Depending on the number of generators and
  the qudit dimension, the stabilizer state could be mixed or pure. In this
  letter, we focus on pure stabilizer state.}\BibitemShut {Stop}%
\end{thebibliography}


\newpage
\section*{Supplementary material for ``Entanglement trimming in stabilizer formalism"}

\begin{appendix}

\section{The qudit stabilizer formalism}

\subsection{Generalized Pauli group for qudit}

The Pauli group defined on qubits plays a vital role in quantum error correction. Here we briefly review the generalization to qudits with dimension $D$. The qubit case is thus just a special case when $D=2$. The Pauli $X$ and $Z$ operator defined on a qudit with $D$ dimension is given by
\begin{equation}
    X=\sum_{j=0}^{D-1}\ket{j}\bra{j+1}, Z=\sum_j\omega^j\ket{j}\bra{j}
\end{equation}
with $\omega=\exp(i2\pi/D)$ and the arithmetic is defined on a modular $D$ ring. Obviously we have $X^D=Z^D=I$ and $XZ=\omega ZX$. For $n$ qudits, we can define a Pauli product 
\begin{equation}
    p(\bm{v})=\omega^{\lambda/2}X_1^{x_1}Z_1^{z_1}...X_n^{x_n}Z_n^{z_n},
\end{equation}
where $\bm{v}=\{x_1,x_2,...;z_1,z_2...\}$ is a $2n$ dimensional vector and $\lambda\in Z_{2D}$. The collection of all possible Pauli products forms a generalized Pauli group $P_n$. For each $p\in P_n$, we have $p^D=I,-I$. It is well known that $-I$ is excluded in any stabilizer code, so only the Pauli product with $P^D=I$ will be included in later discussions. As a remark, when $D$ is prime, the vector representation of the Pauli product corresponds to a symplectic vector space on the field $Z_D$, which provides us a powerful tool to study the generalized Pauli group. In the later section, we give more discussions on Heisenberg-Weyl group in the symplectic space formalism. 

\subsection{Qudit Clifford gates}

The qudit Fourier gate is a generalization of the Hadamard gate, given by
\begin{equation}
    F=\frac{1}{\sqrt{D}}\sum_{j=0}^{D-1}\omega^{jk}\ket{j}\bra{k},
\end{equation}
which satisfies
\begin{equation}
    FZF^\dagger=X,FXF^\dagger=Z^{-1}.
\end{equation}

A multiplicative gate $Q^\alpha$ can be defined for an invertible integer $\alpha\in Z_D$
\begin{equation}
    Q^\alpha=\sum_{j=0}^{D-1}\ket{j}\bra{\alpha j},
\end{equation}
with the relation
\begin{equation}
    Q^\alpha Z {Q^\alpha}^\dagger=Z^\alpha, Q^\alpha X {Q^\alpha}^\dagger=X^{\bar{\alpha}},
\end{equation}
where $\bar{\alpha}$ is the inverse of $\alpha$ in $Z_D$. For qubit case with $D=2$, $\alpha=1$ is the only invertible integer and $Q^{\alpha=1}$ is just an identity operator.

A phase gate for qudit is defined according to the dimension $D$ being even or odd
\begin{equation}
\begin{split}
    W_\text{even}&=\sum_{j=1}^{D-1}\omega^{-j(j+2)/2}\ket{j}\bra{j},\\
    W_\text{odd}&=\sum_{j=1}^{D-1}\omega^{-j(j+1)/2}\ket{j}\bra{j},
\end{split}
\end{equation}
which satisfy
\begin{equation}
    WZW^\dagger=Z,WXW^\dagger=XZ
\end{equation}
for $D$ being odd, and
\begin{equation}
    WZW^\dagger=Z,WXW^\dagger=\omega^{1/2}XZ
\end{equation}
for $D$ being even.

The two-qudit gates include the control-Z (C-Phase) and control-X (CNOT)
\begin{equation}
    \begin{split}
        CZ&=\sum_{j=0}^{D-1}\ket{j}\bra{j}_1\otimes Z_2^j,\\
        CX&=\sum_{j=0}^{D-1}\ket{j}\bra{j}_1\otimes X_2^j.
    \end{split}
\end{equation}
The two gates are related by a Fourier gate on the second qudit $CX=(I_1\otimes F_2)CZ(I_1\otimes F_2)^\dagger$. The CNOT gate satisfies
\begin{equation}
\begin{split}
    CX(I_1\otimes Z_2)CX^\dagger&=Z_1Z_2,\\
    CX(Z_1\otimes I_2)CX^\dagger&=Z_1I_2,\\
    CX(I_1\otimes X_2)CX^\dagger&=I_1X_2,\\
    CX(X_1\otimes I_2)CX^\dagger&=X_1X^{-1}_2
\end{split}
\end{equation}

\subsection{Tripartite stabilizer state}

As described in the main text, a $n$-qudit stabilizer state is defined as 
\begin{equation}
    \rho=\frac{1}{D^n}\sum_{g\in S}g,
\end{equation}
where $S$ is the set of stabilizers. For a prime $D$, we could have exactly $n$ generators, and they can be written in the form
\begin{equation}\label{eqs2}
\begin{split}
    S&=\braket{g_1,g_2,...,g_n}\\
     &=\braket{a_i^{AB}\otimes I^C,I^{AB}\otimes b_j^C,c^{AB}_k\otimes c^C_k},
\end{split}
\end{equation}
with the partition $AB|C$. The first two sets generate the local group $S_{AB},S_C$, while the last sets determines the correlated group $S_{AB|C}$. The generators can be further simplified without changing the entanglement. To show that, the following theorem is needed.
\begin{theorem}
Given any stabilizer generator $p(\bm{v})$ (except the one with identity on the first qudit) on $n$ qudits with prime dimension, it is possible to have a Clifford unitary transformation $U$, such that 
\begin{equation}
    Up(\bm{v})U^\dagger=X_1I_2I_3...I_n
\end{equation}
\end{theorem}
This theorem is discussed in Ref.~\cite{Nielsen2002,Shiang2011}. We give the proof of this theorem here for the purpose of completeness.
\begin{proof}
Given a stabilizer generator in the form $p(\bm{v})=e^{i\phi} X_1^{x_1}Z_1^{z_1}...X_n^{x_n}Z_n^{z_n}$, satisfying $p^D=I$. Knowing the first qudit is not identity, the following transformation for the first qudit can be made: $X_1^{x_1}Z_1^{z_1}\xrightarrow{\text{W gates}}X_1^{x_1}\xrightarrow{Q^\alpha\text{gates}}X_1.$ Note if initially $x_1=0$, we can apply Fourier gate to the first qudit such that $x_1$ becomes nonzero. The above procedure can be repeated for all other qudits if they are not identity already. Then we will get a result $X_1X_2...X_n$. The last step is to apply CNOT gate with first qudit as control until we get $X_1I_1...I_n$. Note there might be a phase factor which can always compensated by applying $Z$ gates. Collecting all gates used gives us the Clifford unitary $U$ for the transformation. 
\end{proof}
This theorem can be directly used to locally transform the stabilizer state in Eq.~(\ref{eqs2}). For example, given the local generator $I^{AB}\otimes b^C_j$ with any $j$, local Cliffords on $C$ can transform it to $I^{AB}\otimes (I...IX)^C$. Notice these local Clifford gates will also change other generators, but no matter how they change, the operator corresponding to the last qudit will either be $I$ or $X$, since all generators must commute. By generator multiplication, the last qudit can be shown as separated $S=\braket{a_i^{AB}\otimes I^{C},I^{AB}\otimes \tilde{b}_j^{C},c^{AB}_k\otimes \tilde{c}^{C}_k}\otimes\braket{X}$, where we use $\tilde{b},\tilde{c}$ to denote the difference upon the transformation. Notice the number of local generators in $I^{AB}\otimes\tilde{b}^C_j$ is now reduced by one. This process can be done repeatedly until the local generators are all separated
\begin{equation}
    S=\braket{a_i^{AB}\otimes I^{C},c^{AB}_k\otimes \tilde{c}^{C}_k}\otimes\braket{X_1X_2...},
\end{equation}
where the generators in the first braket actually defines a new stabilizer state with the system $C$ maximally entangled with $AB$. Since separable qudits doesn't affect the entanglement, the new stabilizer state can be used to study the entanglement structure, as we did in the main text.

\section{Entanglement trimming for square-free dimensional qudits stabilizer state}

As indicated in Ref.~\cite{Shiang2011}, any stabilizer states defined on square-free dimensional qudits can be further decomposed to tensor product of stabilizer states with prime dimension. Specifically, for a stabilizer state $\ket{\Psi}$ with the qudit dimension $D=d_1d_2...d_i...$, we have
\begin{equation}
    \ket{\Psi}=\ket{\psi_1}\otimes\ket{\psi_2}\otimes...\ket{\psi_i}...,
\end{equation}
where $\ket{\psi_i}$ are decomposed stabilizer states with prime dimension $d_i$. Thus for each state $\ket{\psi_i}$, the corresponding stabilizer generators could have a canonical form in terms of any bipartition. This means for any tripartition of the state $\ket{\Psi}^{ABC}$, we can first do entanglement trimming within each decomposed stabilizer state $\ket{\psi_i}^{ABC}$. If the entanglement trimming principle is satisfied for each state, then all $A|BC$ entanglement can be preserved in the system $AB$ while separating the system $C$. 

The situation is more subtle when the condition is violated for some stabilizer states in the tensor product. It might be still possible to realize the entanglement trimming, but non-Clifford operations would be needed, which is out of the range of stabilizer formalism. Although it is not the focus of current paper, we still illustrate it by a simple example: given a three-qudit ($D=6$) stabilizer state, the stabilizer generator is given by
\begin{equation}
   S=\braket{X_AI_BZ_C^3,I_AX_BZ_C^2,Z_A^3Z_B^2X_C}_\text{qudit}.
\end{equation}
Since $D=2\times3$, the above state can be decomposed into a three-qubit and a three-qutrit stabilizer states. Using the mapping: 
\begin{equation}
\begin{split}
&X_{D=6}\rightarrow X_{d_1=2}\otimes X_{d_2=3},\\ &Z_{D=6}\rightarrow Z^{r_1}_{d_1=2}\otimes Z^{r_2}_{d_2=3}, r_i=\frac{D}{d_i}\text{ mod }d_i,
\end{split}
\end{equation}
we can write down their stabilizer generators as
\begin{equation}
\begin{split}
    S=&\braket{X_AI_{B}Z_C,I_AX_BI_C,Z_AI_BX_C}_\text{qubit}\\
    &\otimes\braket{X_AI_BI_C,I_AX_BZ_C,I_AZ_BX_C}_\text{qutrit}.
\end{split}    
\end{equation}
We see each stabilizer state has the generators in the canonical form. The first one contains a Bell state with one ebit $A|C$ entanglement, while the second one contains a EPR state with one etrit $B|C$ entanglement. Thus, the entanglement trimming still can be realized by teleportating the qubit in $C$ to $B$ with the entangled qutrits as a resource. The final state will have one ebit entanglement between the qubit in $A$ and the qutrit in $B$, which is obviously not a stabilizer state any more.

\section{Proof of canonical decomposition of generators for qudit stabilizer state}

We now present the proof of the lemma 2 in the main text by proving the following three lemmas.

\begin{lemma}\label{lm1}
Given a generalized Pauli group $G$ defined on qudits, and denote $C$ as its maximal Abelian subgroup generated by $c(G)$ generators $g_j$, we have
\begin{equation}
    c(G)=\abs{C}\ge\frac{\abs{G}}{2},
\end{equation}
where $c(G)$ is called group compatibility index.
\end{lemma}
\begin{proof}
Denote $\bar{C}$ the subgroup generated by $p=\abs{G}-c(G)$ generators $\bar{g}_j$, s.t. $G=\braket{g_i,
\bar{g}_j}$ ($p$ is called incompatibility index). Then we claim that up to multiplication of elements in generator $C$, each $\bar{g}_j$ can be made to commute with all but one $g_j$ in $C$, which can be seen in two steps: first, it is obvious $\bar{g}_j$ can't commute with all $g_j$ since $C$ itself is maximal Abelian subgroup; second, suppose $\bar{g}_j$ is not commuting with two generator elements $g_1,g_2$ in $C$, e.g., $\bar{g}_jg_1=g_1\bar{g}_j\omega^{\lambda_1},\bar{g}_jg_2=g_2\bar{g}_j\omega^{\lambda_2}.$ By the multiplication of $g_1,g_2$, we can replace $g_2$ by $g_1^mg_2$ with some integer $m$ \cite{Daniel1998}, s.t. $[\bar{g}_j,g_1^mg_2]=0$ \footnote{The existence of integers $m$ is guaranteed when the qudit dimension is prime.}, and $\bar{g}_j$ still not commuting with $g_1$. 

If $c(G)<\abs{G}/2$, by the pigeon hole principle, we could get two generators in $\bar{C}$ which are not commuting with a same element in $C$. With the same method, we could get one of them to be commuting with all elements in $C$. This should be avoided since $C$ being Abelian subgroup is already maximal.
\end{proof}

\begin{lemma}
In the mean time respecting lemma $1$, we can make the generators of $\bar{C}$ commute with each other.
\end{lemma}
\begin{proof}
Respecting lemma \ref{lm1}, the group $G$ can be written as $G=\braket{g_1,...,g_i,...;\bar{g}_1,...,\bar{g}_j,...}$,
where $g_i$ commutes with any generator in $G$ except $\bar{g}_i$ (notice there might be some $g_i$ commuting all generators, which generates the center of the group $G$). We want the same is true for any $\bar{g}_i$. This can be accomplished recursively: suppose we have $[\bar{g}_1,\bar{g}_2]\neq 0$, and we know $[g_2,\bar{g}_2]\neq 0$, then we can replace $\bar{g}_1$ by $\bar{g}^\prime_1=g_2^m\bar{g}_1$ with some integer $m$ such that $[\bar{g}^\prime_1,\bar{g}_2]=0$.
\end{proof}

Now let's return to the qudit stabilizer and focus on the correlated generator $S_{AB}$. Define projection maps: $P_A:g^A\otimes g^B\rightarrow g^A\otimes I^B$ (or $P_B:g^A\otimes g^B\rightarrow I^A\otimes g^B$), where the image of the map forms a subgroup of the generalized Pauli group.

\begin{lemma}
The centers of the group $P_A(S_{AB})$ and $P_B(S_{AB})$ are trivial.  
\end{lemma}

\begin{proof}
Denote $z$ as the rank of the center of the group $P_A(S_{AB})$ and $p$ as its incompatibility index. Now if we look at the stabilizer $S$, we could have $\abs{S_A}+z+p$ independent commuting generators for system $A$. Similarly for system $B$, we have $\abs{S_B}+z+p$ independent commuting generators. The summation of these two sets of generator can not be lager than $n$ (the rank of the stabilizer $S$)
\begin{equation}
    \abs{S_A}+\abs{S_B}+2z+2p\le n.
\end{equation}
However, if we count the number of generators in $S$, we get $ \abs{S_A}+\abs{S_B}+z+2p=n$, which suggests $z=0$. This lemma finally brings the generators of $S$ into its canonical form given in the theorem.
\end{proof}

\section{Entanglement trimming for continuous variable stabilizer state}

The Pauli group for $n$ qubits plays a major role in stabilizer error correction since it forms a basis for qubit operators. For continuous variables, such as $n$ bosonic modes, the Pauli group is called Heisenberg-Weyl (HW) group, composed of all phase space displacements \cite{Richard2004,Kok2010}. Denote $q_i$ and $p_i$ as the $i_\text{th}$ mode position and momentum quadratures, which satisfy the canonical commutation relation $[\hat{q}_i,\hat{p}_i]=i$. The HW group $\mathcal{H}_n$ is a Lie group and those quadratures are the infinitesimal generators, which forms the Lie algebra. In fact, tracking the evolution of the Lie algebra elements could lead us to a nice proof of continuous version of Gottesman-Knill theorem \cite{Stephen2002}. The HW group element can be written as
\begin{equation}
    g(\bm{v})=\exp\left({i\sqrt{2\pi}\sum_i^ns_i\hat{p}_i+t_i\hat{q}_i}\right),
\end{equation}
where we immediately see the one-to-one correspondence between the element $g(\bm{v})$ and a vector $\bm{v}=(s_1,t_1,s_2,t_2,...,s_n,t_n)\in \mathcal{R}^{2n}$. The group elements satisfy
\begin{equation}
    g(\bm{v})g(\bm{v}^\prime)=g(\bm{v}^\prime)g(\bm{v})e^{i2\pi\sigma(\bm{v},\bm{v}^\prime)},
\end{equation}
where $\sigma(\bm{v},\bm{v}^\prime)=\sum_i^n(s_it_i^\prime-s_i^\prime t_i)$ is the standard symplectic form on the symplectic vector space $(\mathcal{R}^{2n},\sigma)$. This correspondence enables us to define Abelian subgroups in terms of the symplectic vectors
\begin{equation}
\mathcal{S}=\{g(\textbf{u})\in\mathcal{H}_n|\textbf{u}=\sum_{i=1}^K\alpha_i\bm{v}_i,\alpha_i\in\mathcal{R}\},
\end{equation}
where $\sigma(\bm{v}_i,\bm{v}_j)=0$ for all $1\le i,j\le K\le n$ \footnote{Note this constraint is much stronger than the general requirement of an Abelian subgroup, where the skew-product only needs to be an interger \cite{Gkp2001,Jim2001}.}. Thus we also have a corresponding vector space for the Abelian subgroup $\mathcal{S}$, which is $\mathcal{V_S}=span\{\bm{v}_1,\bm{v}_2,...,\bm{v}_K \}$ (Note $\mathcal{V_S}$ is in general not a symplectic vector space). The Abelian subgroup could form a set of stabilizers and when $K=n$, the stabilizers determine a stabilizer state \cite{Jingtao2009}. Similar to qubit (qudit) case with bipartition $A|B$, the stabilizer generator contains local ($\mathcal{S}_A,\mathcal{S}_B$) and correlated ($\mathcal{S}_{AB}$) subsets. The canonical form of the generators can be expressed in terms of the corresponding vectors.

\subsection{Canonical decomposition of generators for continuous variable stabilizer state}

\begin{theorem}---\textbf{Canonical form of stabilizer generators for continuous variables}:
Given an $n$-mode stabilizer state with bi-partition $A$ and $B$, the vector correspondence $\mathcal{V_S}$ to its stabilizer group $\mathcal{S}$ can always be brought into the following canonical form 
\begin{equation}
    \mathcal{V_S}=span\left\{\bm{v}_i^A\cup\textbf{0}^B,\textbf{0}^A\cup\bm{v}_j^B,\textbf{u}_k^A\cup\textbf{u}_k^B, \bar{\textbf{u}}_k^A\cup\bar{\textbf{u}}_k^B \right\},
\end{equation}
where the first two subsets correspond to the local generators for $\mathcal{S}_A$ and $
\mathcal{S}_B$, while the last two paired subsets correspond to correlated generators for $\mathcal{S}_{AB}$. Those paired vectors satisfy $\sigma(\textbf{u}_{k_i}^A,\bar{\textbf{u}}_{k_j}^A)=-\sigma(\textbf{u}_{k_i}^B,\bar{\textbf{u}}_{k_j}^B)\neq 0$ when $k_i=k_j$, and satisfy $\sigma(\textbf{u}_{k_i}^A,\bar{\textbf{u}}_{k_j}^A)=\sigma(\textbf{u}_{k_i}^B,\bar{\textbf{u}}_{k_j}^B)= 0$ when $k_i\neq k_j$.
\end{theorem}
 The above theorem can be proved by exploring the structure of the symplectic vector space, and we present it in terms of the following lemmas.

\begin{lemma}\label{cvle1}
Denote $V_n$ as a vector space (not necessarily a symplectic space) and $V_S= span\{\bm{v}_i\}$ its maximal subspace which satisfies ${V_S}={V_S}^\sigma$, and take ${\bar{V}_S}= span\{\bar{\bm{v}}_j\}$ as another subspace of $V_n$ that satisfies ${V}_n=span\{{V_S};{\bar{V}_S}\}$, where $\sigma$ is the standard symplectic form defined on any enlarged symplectic vector space, $\{\bm{v}_i\}$ and $\{\bar{\bm{v}}_j\}$ are the linear independent sets correspondingly. Then we will have $dim({V_S})\ge dim({\bar{V}_S})$; the vectors satisfy $\sigma(\bm{v}_i,\bar{\bm{v}}_j)=\alpha\delta_{ij} \forall i,j$, $\alpha\in\mathcal{R}$ and $\sigma(\bar{\bm{v}}_i,\bar{\bm{v}}_j)=0, \forall i,j$.
\end{lemma}

\begin{proof}
 We first prove that each vector $\bar{\bm{v}}_j$ in the set $\{\bar{\bm{v}}_j\}$ can be made to symplectic orthogonal to all but one of $\bm{v}_i$ in ${V_S}$. The reason is: first, $\forall \bar{\bm{v}}_j$ can't be symplectic orthogonal to all elements in ${V_S}$ since it is already maximal, meaning no more linear independent vectors can be added; second, if $\bar{\bm{v}}_j$ is not symplectic orthogonal to two elements $\bm{v}_i,\bm{v}_k\in {V_S}$, e.g., $\sigma(\bm{v}_i,\bar{\bm{v}}_j)=\alpha, \sigma(\bm{v}_k,\bar{\bm{v}}_j)=b$, then we can define a new vector $\bm{v}_k^\prime=\frac{\alpha}{b}\bm{v}_k-\bm{v}_i$ to replace $\bm{v}_k$ in the set, such that $\sigma(\bm{v}_k^\prime,\bar{\bm{v}}_j)=0$. Repeat this operation until we have $\sigma(\bm{v}_i,\bar{\bm{v}}_j)=\alpha\delta_{ij}, \forall i,j$.
 
We prove $dim (V_S)\ge dim(\bar{V}_S)$ by contradiction. If we have $dim ({V_S}) < dim({\bar{V}_S})$, then by the pigeon hole principle we can find $i_1\neq i_2$ such that $\bar{\bm{v}}_{i_1}$ and $\bar{\bm{v}}_{i_2}$ are not symplectic orthogonal to a same $\bm{v}_j$. Similarly we can replace $\bar{\bm{v}}_{i_1}$ or $\bar{\bm{v}}_{i_2}$ by a new vector that is symplectic orthogonal to $\bm{v}_j$, making ${V_S}$ not the maximal mutually symplectic orthogonal subspace. 

Thus, the vector space $V_n$ can be expressed as (up to relabel) ${V}_n= span\{\bm{v}_1,...,\bm{v}_{dim({V_S})};\bar{\bm{v}}_1,...,\bar{\bm{v}}_{dim({\bar{V}_S})}\}$,
where ($\bm{v}_j$,$\bar{\bm{v}}_j$) form a pair with $\bm{v}_j$ being symplectic orthogonal to all vectors but $\bar{\bm{v}}_j$. Now we prove the same is true for $\bar{\bm{v}}_j$, where we need to show $\sigma(\bar{\bm{v}}_i,\bar{\bm{v}}_j)=0, \forall i,j$. This can be done recursively: suppose there are two vectors $\bar{\bm{v}}_{j_1}$ and $\bar{\bm{v}}_{j_2}$ which are not symplectic orthogonal, and we also know $\bar{\bm{v}}_{j_1}$ is not symplectic orthogonal to $\bm{v}_{j_1}$, thus we can replace $\bar{\bm{v}}_{j_2}$ by a similar new vector that is symplectic orthogonal to $\bar{\bm{v}}_{j_1}$.
\end{proof}

\begin{lemma}
Given a symplectic vector space $(\mathcal{R}^{2n},\sigma)$ corresponding to a HW group $\mathcal{H}_n$. A pair of transverse Lagrangian planes $(L,\bar{L})$, e.g., $L=\text{span}\{\textbf{u}_i|1\le i\le n\}$ and $\bar{L}$=span$\{\bar{\textbf{u}}_j|1\le j\le n\}$ define two maximal Abelian subgroups: $\mathcal{H}_L$ with generators $g(\textbf{u}_i)$ and $\mathcal{H}_{\bar{L}}$ with generators $g(\bar{\textbf{u}}_j)$. The corresponding vectors satisfy $\sigma(\textbf{u}_i,\bar{\textbf{u}}_j)=\alpha\delta_{ij},\alpha\in\mathcal{R}, \forall i,j$.
\end{lemma}

\begin{proof}
Inspecting the previous lemma, if we require $V_n=(\mathcal{R}^{2n},\sigma)$ as a symplectic vector space, we could get $dim(V_S)=dim(\bar{V}_S)=n/2$. Since $V_S\cap\bar{V}_S=\emptyset$, they define a pair of transverse Lagrangian plane, denoted as $L= span\{\bm{u}_i|1\le i\le n\}$ and $\bar{L}= span\{\bar{\bm{u}}_j|1\le j\le n\}$. The commutation relation of those vectors just follow.
\end{proof}

 Notice the set $\{\bm{u}_i;\bar{\bm{u}}_j\}$ doesn't form a symplectic basis for $(\mathcal{R}^{2n},\sigma)$ since we just have $\sigma(\bm{u}_i,\bar{\bm{u}}_j)=\alpha\delta_{ij}\neq\delta_{ij}$. However, we can easily promote it to be a symplectic basis by normalizing those vectors to get $\alpha=1$, which completes the operation indicated in the renowned symplectic Gram-Schimdt theorem \cite{Gosson2006}. As an interesting observation, the Lagrangian plane of a standard symplectic vector space actually defines a continuous variable stabilizer state, suggesting more interesting structure of stabilizer state can be revealed in the framework of symplectic geometry. 
 
 To present the next lemma that finally proves the theorem, let's first denote $\mathcal{V}_A=span\{\bm{v}^A_i\cup\bm{0}^B\}$, $\mathcal{V}_B=span\{\bm{0}^A\cup\bm{v}_j^B\}$ and $\mathcal{V}_{AB}=span\{(\bm{v}^{A}\cup\bm{v}^B)_k\}$ as the vector space corresponding to the generator sets $\mathcal{S}_A$, $\mathcal{S}_B$ and $\mathcal{S}_{AB}$, respectively, and define a projection map $f_A(\mathcal{V}_{AB}):(\bm{v}^A\cup\bm{v}^B)_k\rightarrow\bm{v}_k^A\cup\bm{0}^B$ (similarly we have another map $f_B(\mathcal{V}_{AB}):(\bm{v}^A\cup\bm{v}^B)_k\rightarrow\bm{0}^A\cup\bm{v}_k^B$). Notice the image of the map $f_A(\mathcal{V}_{AB})$ is linearly independent with the vectors defined in $\mathcal{V}_A$. 

\begin{lemma}
The image of the map $f_A(\mathcal{V}_{AB})$ $(\text{or }f_B(\mathcal{V}_{AB}))$ is a symplectic vector space.
\end{lemma}

\begin{proof}
Obviously, $f_A(\mathcal{V}_{AB})= span\{\bm{v}^A_k\}$ is a vector space and it is bilinear and alternating in terms of the standard symplectic form $\sigma$. We only have to show the vector space $f_A(\mathcal{V}_{AB})$ is nondegenerate. Suppose $f_A(\mathcal{V}_{AB})=span\left\{ \{\bm{v}_i\}_{i=1,...,p};\{\bar{\bm{v}}_j\}_{j=1,...,\bar{p}} \right\}$ with $p-\bar{p}=t\ge 0$, meaning there are $t$ non-zero vectors that are symplectic orthogonal to all vectors in the image. Now if we look at the vector space $\mathcal{V_S}$, we can find $dim(\mathcal{V}_A)+p=dim(\mathcal{V}_A)+t+\bar{p}$ independent vectors on system $A$. Similarly considering $f_B(\mathcal{V}_{AB})$, we could have $dim(\mathcal{V}_B)+t+\bar{p}$ independent vectors on system $B$. The summation of these independent vectors can not be larger than $n$, which is the dimension of $\mathcal{V_S}$
\begin{equation}
    dim(\mathcal{V}_A)+dim(\mathcal{V}_B)+2t+2\bar{p}\le n.
\end{equation}
However, we know the independent vectors on $\mathcal{V_S}$ is simply counted as $dim(\mathcal{V}_A)+dim(\mathcal{V}_B)+t+2\bar{p}=n$. The only possibility is $t=0$, which means all vectors in $f_A(\mathcal{V}_{AB})$ are paired, thus only zero vectors can be symplectic orthogonal to all vectors in it, leading to nondegeneracy. The canonical form given in the theorem just follows.
\end{proof}

As a remark, all symplectic analysis for continuous variable stabilizer state apply naturally to qudit with prime dimension $D$, where the symplectic vector space is defined in the field $Z_D$. With the canonical decomposition of stabilizer generators, we can analogously construct ``the bigger man principle" for continuous variable stabilizer state.
\begin{theorem}
---\textbf{the bigger man principle for continuous variable}: for a continuous variable stabilizer state with tri-partition $AB|C$, according to the canonical decomposition theorem, the vector space corresponding to its stabilizer can be written in the canonical form,
\begin{equation}
    \mathcal{V_S}=span\left\{\bm{v}_i^{AB}\cup\textbf{0}^C,\textbf{0}^{AB}\cup\bm{v}_j^C,\textbf{u}_k^{AB}\cup\textbf{u}_k^C, \bar{\textbf{u}}_k^{AB}\cup\bar{\textbf{u}}_k^C \right\}.
\end{equation}
The $A|BC$ entanglement can be transformed to $A|B$ with local operation on $C$, if and only if, (1) $\forall$ $k$, either $\textbf{u}_k^{AB}=\textbf{0}^A\cup\textbf{u}^B_k$ or $\bar{\textbf{u}}_k^{AB}=\textbf{0}^A\cup\bar{\textbf{u}}^B_k$, or (2) even pairs of correlated vectors come in the form $\{ \textbf{0}^A\textbf{u}_k^{B}\cup\textbf{u}_k^C, \textbf{0}^A\bar{\textbf{u}}_k^{B}\cup\bar{\textbf{u}}_k^C,\textbf{u}_{k^\prime}^{A}\textbf{0}^B\cup\textbf{u}_{k^\prime}^C, \bar{\textbf{u}}_{k^\prime}^{A}\textbf{0}^B\cup\bar{\textbf{u}}_{k^\prime}^C\}$.
\end{theorem}

\end{appendix}

\end{document}